\documentclass[10pt,conference]{IEEEtran}
\usepackage{amsmath,amssymb}
\usepackage{graphicx}
\usepackage{mathrsfs}
\usepackage{tikz}
\usepackage{booktabs}
\usepackage{enumerate}
\usepackage{psfrag}	
\usepackage{array}
\usepackage{multirow,hhline}
\usepackage{exscale}
\usepackage{color}
\usepackage{colortbl}
\usepackage{epsfig,subfigure}
\usepackage{cite}
\usepackage{amsthm}
\renewcommand{\vec}[1]{\ensuremath{\mathbf{#1}}}
\newcommand{\alert}[1]{\textcolor{black}{#1}}
\newtheorem{theorem}{Theorem}
\newtheorem{lemma}{Lemma}

\begin{document}

\IEEEoverridecommandlockouts
\title{The Degrees-of-Freedom of Multi-way Device-to-Device Communications is Limited by 2}
\author{
\IEEEauthorblockN{Anas Chaaban$^{*}$, Henning Maier$^{\dagger}$, Aydin Sezgin$^{*}$}
\IEEEauthorblockA{$^{*}$Institute of Digital Communication Systems, Ruhr-Universit\"at Bochum (RUB), Germany\\
Email: {anas.chaaban,aydin.sezgin@rub.de}\\
$^{\dagger}$Lehrstuhl f\"ur Theoretische Informationstechnik, RWTH Aachen, Germany,\\
Email: {maier@ti.rwth-aachen.de}
}
\thanks{
The work of A. Chaaban and A. Sezgin is supported by the German Research Foundation, Deutsche Forschungsgemeinschaft (DFG), Germany, under grant SE 1697/5.

The work of H. Maier is supported by the DFG  under grant PACIA-Ma 1184/15-2 and by the UMIC Research Centre, RWTH Aachen University. 
}
}

\maketitle

% \footnotetext[1]{a}

\begin{abstract}
A 3-user device-to-device (D2D) communications scenario is studied where each user wants to send and receive a message from each other user. This scenario resembles a 3-way communication channel. The capacity of this channel is unknown in general. In this paper, a sum-capacity upper bound that characterizes the degrees-of-freedom of the channel is derived by using genie-aided arguments. It is further shown that the derived upper bound is achievable within a gap of 2 bits, thus leading to an approximate sum-capacity characterization for the 3-way channel. As a by-product, interesting analogies between multi-way communications and multi-way relay communications are concluded.
\end{abstract}
% \begin{IEEEkeywords}
% Y-Channel, 
% \end{IEEEkeywords}

\section{Introduction}
The increase in the demand on data-rates in communication networks accompanied by the spectrum shortage has motivated researchers to seek new methods to combat these challenges. Several solutions have been proposed to overcome these challenges. Among the most promising solutions proposed to alleviate the limitations of the existing communication networks is a communication mode known as Device-to-device (D2D) communication.

D2D communication is defined as direct communication between mobile nodes in close proximity without incorporating a base station in the process (see \cite{AsadiWangMancuso} for a survey about D2D communications). It was first introduced in \cite{LinHsu}. Realizing that D2D communication helps offloading some data traffic from the cellular network, standardization bodies \cite{3GPPTR22803} are considering it as a potential component of future communication systems. Due to this fact, the research focus on D2D communication has increased recently. For instance, the potential of interference alignment \cite{Jafar} for D2D networks has been studied in \cite{ElkotbyElsayedIsmail}. Furthermore, a new spectrum sharing mechanism for D2D communication \alert{that outperforms state-of-the-art spectrum sharing mechanisms} has been proposed in \cite{NaderializadehAvestimehr}. \alert{At each time, this mechanism (which can also be implemented in a distributed way) schedules users that can achieve a near-optimal performance by communicating simultaneously while treating interference as noise.}

A D2D communication scenario between two nodes can be modelled by a two-way channel (TWC), where each node can communicate with the other node simultaneously. The TWC was introduced by Shannon in \cite{Shannon_TWC} where capacity upper and lower bounds were derived, and the capacity of some classes of TWC was characterized. If more than two nodes want to establish D2D communications, the scenario can be modelled by a multi-way channel. However, the extension of the TWC to more than two users has not been studied so far. Hence, information-theoretic results on the fundamental limits of communications over a multi-way channel are not available to date. The emergence of D2D communications calls upon establishing these limits. This direction is pursued in this paper.

We study a three-way channel consisting of three full-duplex nodes that want to communicate pair-wise simultaneously with each other. We call the resulting network the 3-D2D channel. We consider a Gaussian channel, i.e., where all nodes are disturbed by a Gaussian noise. Furthermore, we consider reciprocal channels where the channel gain between two nodes is the same in both directions. For this network, we derive novel upper bounds on the achievable rates which characterize the degrees-of-freedom (DoF) of the channel. Interestingly, while the cut-set bounds characterize the DoF of the TWC \cite{Han}, they do not characterize the DoF of the 3-D2D network, which is in turn characterized by our new bound. Furthermore, we prove the achievability of this sum-capacity upper bound within a gap of 2 bits. It turns out that the simple opportunistic strategy of letting the two users sharing the strongest channel communicate while leaving the third user silent suffices to achieve the sum-capacity of the 3-D2D channel within a constant gap. 

As a by-product of this characterization, we obtain the following conclusion. Contrary to some networks where the DoF changes by increasing the number of users (such as the interference channel \cite{CadambeJafar_KUserIC}), the DoF of the multi-way channel does not change if we increase the number of users from two to three. Namely, the sum-capacity of the 3-D2D channel scales (as $\mathsf{SNR}$ increases) as twice the capacity of the strongest channel between two users (2 DoF). This behaviour is the same as that in the TWC \cite{Han}.

Furthermore, we discover the following analogies between multi-way communication (as in the 3-D2D channel) and multi-way relaying. It was observed in \cite{ChaabanSezginAvestimehr_YC_SC} that the DoF of the two-way relay channel does not change if we increase the number of users. As mentioned above, the same observation holds for multi-way communications when going from the 2-user to the 3-user case. Furthermore, it was observed in \cite{ChaabanSezginAvestimehr_YC_SC} that independent of the number of users, the optimal DoF can be achieved by letting the two strongest users communicate while leaving the remaining users silent. The same holds for the 3-D2D channel where the optimal DoF can be achieved by letting the two users sharing the strongest channel communicate. These analogies are interesting, and motivate the search for further analogies between the two types of networks.

The paper is organized as follows. The system model of the 3-D2D channel is described in Section \ref{Model}. The main result of the paper is given in Section \ref{MainResults}. Upper bounds on the capacity of the channel are given in Section \ref{UpperBounds} and transmission strategies are given in Section \ref{Achievability}. Finally, we conclude the paper with Section \ref{Conclusion}. Throughout the paper, we use $x^{n}$ to denote the length-$n$ sequence $(x(1),\cdots,x(n))$ and \alert{we use capital letters to denote random variables.} The function $C(x)$ is used to denote $\frac{1}{2}\log(1+x)$ for $x\geq0$, .

\section{System Model}
\label{Model}
The 3-D2D channel is a multi-way channel consisting of three users communicating simultaneously with each other. The channel is fully connected as shown in Fig. \ref{DChannelModel} and all nodes are full-duplex. Each user in the 3-D2D channel has two independent messages, one for each remaining user. That is, user 1 has messages $m_{12}$ and $m_{13}$ intended to users 2 and 3, respectively. Similarly user 2 has $m_{21}$ and $m_{23}$, and user~3 has $m_{31}$ and $m_{32}$. The message $m_{jk}$ is chosen uniformly from a message set $\mathcal{M}_{jk}=\{1,\cdots,2^{nR_{jk}}\}$, where $R_{jk}$ is the rate of the message and $n$ is the code length.

\begin{figure}[t]
\centering
\includegraphics[width=\columnwidth]{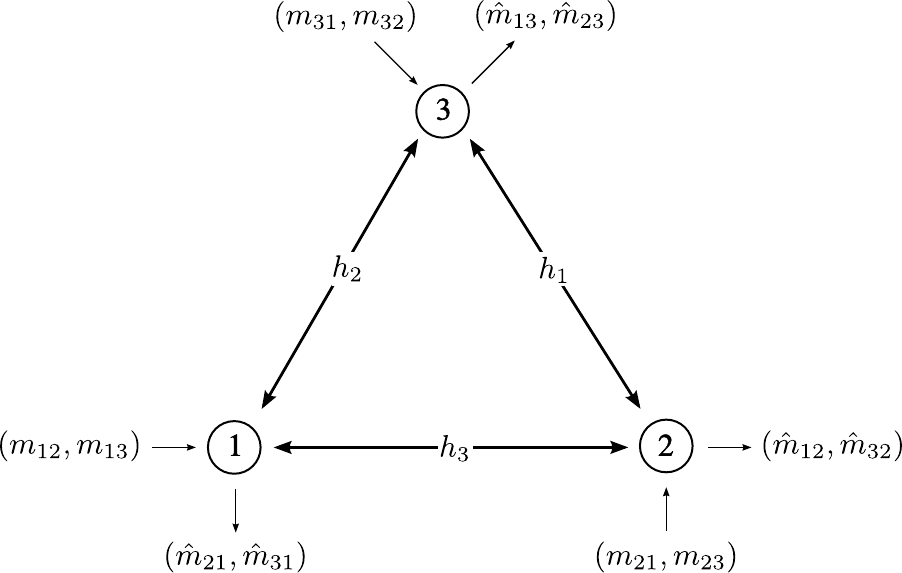}
\caption{Multi-way communication over the reciprocal 3-D2D channel. Each user sends a message to and receives a message from the two other users.}
\label{DChannelModel}
\end{figure}

To send his messages, user 1 sends a transmit signal $x_1^{n}$ of length $n$ symbols, whose $i$th symbol $x_1(i)\in\mathbb{R}$ is constructed from the messages $m_{12}$ and $m_{13}$ and the received symbols at user 1 up to time instant $i$, i.e., $y_1^{i-1}$, using an encoding function $\mathcal{E}_{1,i}$. Users 2 and 3 construct their signals similarly. Thus, we may write for each user $j$
\begin{align}
\label{Encoding}
x_j(i)=\mathcal{E}_{j,i}(m_{jk},m_{j\ell},y_j^{i-1}),
\end{align}
for distinct $j,k,\ell\in\{1,2,3\}$. The received signals are given by
\begin{align}
y_1(i)=h_3x_2(i)+h_2x_3(i)+z_1(i),\\
y_2(i)=h_3x_1(i)+h_1x_3(i)+z_2(i),\\
y_3(i)=h_2x_1(i)+h_1x_2(i)+z_3(i),
\end{align}
where $h_3,h_2,h_1\in\mathbb{R}$ are the (globally known) real-valued (static) channel coefficients, and $z_1,z_2,z_3\in\mathbb{R}$ represent the independent noises at users 1, 2, and 3, respectively, which are Gaussian with unit variance and i.i.d. over time. Note that the channels are assumed to be reciprocal, i.e., the channel gain between two users is the same in both directions. We assume without loss of generality that 
\begin{align}
\label{Ordering}
h_3^2\geq h_2^2\geq h_1^2,
\end{align} 
i.e., users 1 and 2 share the strongest channel. Each user has a power constraint $P$, i.e.,
\begin{align}
\sum_{i=1}^{n}\mathbb{E}[X_j(i)^{2}]\leq nP,\quad j\in\{1,2,3\}.
\end{align}

After receiving $y_j^{n}$, user $j$ decodes his desired messages $\hat{m}_{kj}$ and $\hat{m}_{\ell j}$ ($\{j,k,\ell\}=\{1,2,3\}$) using a decoding function $\mathcal{D}_j$ and his own messages $m_{jk}$ and $m_{j\ell}$, i.e., 
\begin{align}
\label{Decoder}
(\hat{m}_{kj},\hat{m}_{\ell j})=\mathcal{D}_j(y_j^{n},m_{jk},m_{j\ell}).
\end{align} 
An error occurs if $m_{jk}\neq \hat{m}_{jk}$ for some $j\neq k$. The collection of message sets, encoders, and decoders defines a code for the 3-D2D channel denoted $(n,\vec{R})$ where $$\vec{R}=(R_{12},R_{13},R_{21},R_{23},R_{31},R_{32}),$$
and induces an error probability $P_{e,n}$.

A rate tuple $\vec{R}$ is said to be achievable if there exist a sequence of $(n,\vec{R})$ codes such that the $P_{e,n}$ can be made arbitrarily small by increasing $n$. In this paper, we are interested in the sum-capacity $C_\Sigma$ of the 3-D2D channel defined as the maximum achievable sum-rate $R_\Sigma=\sum_j\sum_{k\neq j}R_{jk}$. The main result of the paper is given in the following section.

\section{Main Result}
\label{MainResults}
The main result regarding the sum-capacity of the considered 3-D2D channel is presented in the following theorem.
\begin{theorem}
\label{Thm:Main}
The sum-capacity of the 3-D2D channel is bounded by
\begin{align}
\label{ApproxCap}
2C(h_3^{2}P)\leq C_\Sigma \leq 2C(h_3^{2}P)+2.
\end{align}
\end{theorem}
This theorem provides an approximate characterization of the sum-capacity of the given network within a gap of 2 bits. The proof of the theorem is given in Sections \ref{UpperBounds} where a novel genie-aided upper bound is derived leading to the right-hand side of \eqref{ApproxCap}, and in Section \ref{Achievability} where the achievability of the upper bound within a gap of 2 bits is shown leading to the left-hand side of \eqref{ApproxCap}.

The following interesting conclusions can be drawn from this theorem. First, the sum-capacity of the 3-D2D channel has the same scaling behaviour as that of the two-way channel~\cite{Han}. That is, both the 2-user case and the 3-user case have 2 DoF, where the DoF is defined as
\begin{align}
\text{DoF}=\lim_{\mathsf{SNR}\to\infty}\frac{C_\Sigma(\mathsf{SNR})}{\frac{1}{2}\log(\mathsf{SNR})}.
\end{align}
Thus, in a multi-way channel, the DoF stays constant at 2 if we increase the number of users from 2 to 3 (contrary to some channels such as the interference channel \cite{CadambeJafar_KUserIC} where the DoF depends on the number of users). Furthermore, the sum-capacity can be approached within a constant gap by letting the two users sharing the strongest channel (users 1 and 2) communicate while keeping the remaining user silent. That is, by letting users 1 and 2 communicate while keeping user 3 silent, we can achieve a sum-rate which is within 2 bits (at most) of the sum capacity. Interestingly, the same observations were concluded for multi-way relay communications in \cite{ChaabanSezginAvestimehr_YC_SC} where it was concluded that extending the two-way relay channel \cite{WilsonNarayananPfisterSprintson, AvestimehrSezginTse} to a three-way relay channel (Y-channel \cite{LeeLimChun,ChaabanSezgin_ISIT12_Y,ChaabanOchsSezgin}) preserves the same DoF, and that this scaling can be achieved by letting the two strongest users communicate while keeping the third user silent.

Next, we present new upper bounds for the 3-D2D channel which are necessary for proving Theorem \ref{Thm:Main}.

\section{Upper bounds}
\label{UpperBounds}
The cut-set bounds can be used to obtain an upper bound on the achievable rates in a multi-way channel. In fact, the cut-set bounds are tight for the 2-user case (the two-way channel) as shown in \cite{Han}. However, this is not the case for the 3-user case, i.e., the 3-D2D channel. For instance, in the 3-D2D channel, the rates of the messages from and to user 1 can be bounded by the cut-set bounds as
\begin{align}
R_{12}+R_{13}&\leq I(X_1;Y_2,Y_3|X_2,X_3)\\
R_{21}+R_{31}&\leq I(X_2,X_3;Y_1|X_1),
\end{align}
for some input distribution $p(x_1,x_2,x_3)$, where $X_j$ and $Y_j$ denote the input and output random variables. Similar bounds can be obtained for the other 2 users. By maximizing these bounds using the Gaussian input distribution, we can write
\begin{align}
R_{12}+R_{13}&\leq C(h_3^{2}P+h_2^{2}P)\\
R_{21}+R_{31}&\leq C(h_3^{2}P+h_2^{2}P).
\end{align}
Note that these bounds scale as $\frac{1}{2}\log(P)$ as $P$ grows. Similarly, we can bound the rates $R_{21}+R_{23}$,  $R_{12}+R_{32}$,  $R_{31}+R_{32}$, and $R_{13}+R_{23}$ with quantities that have the same scaling behaviour, i.e., they scale as $\frac{1}{2}\log(P)$. Note that this leads to a DoF of 3. It turns out however that the cut-set bounds are not tight for this network. In what follows, we provide bounds on the achievable rates that lead to a tighter sum-capacity upper bound.

\begin{lemma}
\label{Lemma:FirstBound}
An achievable rate for the 3-D2D channel must satisfy
\begin{align}
R_{21}+R_{31}+R_{32} &\leq C(h_3^2P+h_2^2P)+C\left(\frac{h_1^{2}}{h_2^{2}}\right).
\end{align}
\end{lemma}
\begin{proof}
The intuition for finding this bound is as follows. User 1 can decode $m_{21}$ and $m_{31}$ from $y_1^{n}$, $m_{12}$, and $m_{13}$ \eqref{Decoder}. We would like to enable user 1 to decode one more message (here $m_{32}$) by giving it the least possible side information. To guarantee this, we need to enable user 1 to construct $y_2^n$. Now note that if we give $m_{23}$ to user 1 as side information, then after decoding $m_{21}$, user 1 has all the information available to user 2 at time instant $i=1$, i.e., the message pair $(m_{21},m_{23})$. Thus, user 1 can generate the first symbol of $x_2^{n}$, i.e., $x_2(1)$ (cf. \eqref{Encoding}). Now by using $y_1(1)=h_3x_2(1)+h_2x_3(1)+z_1(1)$, user 1 can obtain $\tilde{y}_1(1)=h_2x_3(1)+z_1(1)$. After multiplying $\tilde{y}_1(1)$ by $\frac{h_1}{h_2}$, and  adding $h_3x_1(1)$, user 1 obtains 
\begin{align}
\tilde{y}_2(1)=h_3x_1(1)+h_1x_3(1)+\frac{h_1}{h_2}z_1(1).
\end{align}
This is a less noisy version of $y_2(1)$. Note that in order to repeat this procedure for the following time instances ($i>1$), it is not enough to have $\tilde{y}_2(1)$. Rather, we need $y_2(1)$ exactly in order to produce $x_2(2)$ (which is generated from $m_{21}$, $m_{23}$, and $y_2(1)$ as in \eqref{Encoding}) which is essential for obtaining $y_2(2)$. In order to obtain $y_2(1)$ at user 1, we give him the signal $\tilde{z}_2^n=z_2^n-\frac{h_1}{h_2}z_1^n$. Now by adding $\tilde{y}_2(1)$ to $\tilde{z}_2(1)$, user 1 gets $y_2(1)$ and can generate $x_2(2)$. By repeating this procedure, user 1 can generate all instances of $y_2^n$. Then using the decoded $m_{21}$, the provided side information $m_{23}$, and the generated $y_2^{n}$, user 1 can decode $m_{32}$ just as user 2 can decode it. Therefore, we can write 
\begin{align*}
n(R_{21}+R_{31}+R_{32}-\varepsilon_n)\leq I(\hat{\vec{M}}_{1},M_{32};Y_1^n,\tilde{Z}_2^n, \vec{M}_1,M_{23})
\end{align*}
by using Fano's inequality, where $\varepsilon_n\to0$ as $n\to\infty$, and where we used $\vec{M}_{1}$ and $\hat{\vec{M}}_{1}$ to denote the random vectors $(M_{12},M_{13})$ and $(M_{21},M_{31})$ indicating the message pairs sent and received at user 1, respectively. This bound can be written as
\begin{align}
& n(R_{21}+R_{31}+R_{32}-\varepsilon_n)\nonumber\\
& \stackrel{(a)}{\leq} I(\hat{\vec{M}}_{1},M_{32};Y_1^n,\tilde{Z}_2^n| \vec{M}_{1},M_{23})\\
& \stackrel{(b)}{=} \sum_{i=1}^{n} I(\hat{\vec{M}}_{1},M_{32};Y_1(i),\tilde{Z}_2(i)|\vec{M}_{1},M_{23},Y_1^{i-1},\tilde{Z}_2^{i-1})\nonumber\\
& \stackrel{(b)}{=} \sum_{i=1}^{n} I(\hat{\vec{M}}_{1},M_{32};Y_1(i)|\vec{M}_{1},M_{23},Y_1^{i-1},\tilde{Z}_2^{i-1})\nonumber\\
\label{TIE}
&\quad +\sum_{i=1}^{n} I(\hat{\vec{M}}_{1},M_{32};\tilde{Z}_2(i)| \vec{M}_{1},M_{23},Y_1^{i},\tilde{Z}_2^{i-1}),
\end{align}
where $(a)$ follows from the independence of the messages, and $(b)$ follows by using the chain rule. The first mutual information expression in \eqref{TIE} can be bounded as
\begin{align}
&I(\hat{\vec{M}}_{1},M_{32};Y_1(i)|\vec{M}_{1},M_{23},Y_1^{i-1},\tilde{Z}_2^{i-1})\nonumber\\
&\stackrel{(c)}{=}
h(Y_1(i)|\vec{M}_{1},M_{23},Y_1^{i-1},\tilde{Z}_2^{i-1})-h(Y_1(i)|\vec{M},Y_1^{i-1},\tilde{Z}_2^{i-1})\nonumber\\
&\stackrel{(d)}{\leq} h(Y_1(i))-h(Y_1(i)|\vec{M},Y_1^{i-1},\tilde{Z}_2^{i-1},Z_2^{i-1},Z_3^{i-1})\\
&\stackrel{(e)}{=} h(Y_1(i))-h(Z_1(i)|\vec{M},Z_1^{i-1},\tilde{Z}_2^{i-1},Z_2^{i-1},Z_3^{i-1})\\
&= h(Y_1(i))-h(Z_1(i)|\vec{M},Z_1^{i-1},Z_2^{i-1},Z_3^{i-1})\\
\label{IE1}
&\stackrel{(f)}{=} h(Y_1(i))-h(Z_1(i)),
\end{align}
where $(c)$ follows by using the definition of mutual information and by defining $\vec{M}=(M_{12},M_{13},M_{21},M_{23},M_{31},M_{32})$, $(d)$ follows 
since conditioning does not increase entropy, $(e)$ follows since by knowing $\vec{M}$, $Z_2^{i-1}$, $Z_3^{i-1}$, and $Y_1^{i-1}$, all random variables required to determine $X_2^{i}$, and $X_3^{i}$ are given, and $(f)$ follows from the independence of the messages and the noise random variables. The second mutual information expression in \eqref{TIE} can be bounded as
\begin{align}
&I(\hat{\vec{M}}_{1},M_{32};\tilde{Z}_2(i)| \vec{M}_{1},M_{23},Y_1^{i},\tilde{Z}_2^{i-1})\nonumber\\
&= h(\tilde{Z}_2(i)|\vec{M}_{1},M_{23},Y_1^{i},\tilde{Z}_2^{i-1})-h(\tilde{Z}_2(i)| \vec{M},Y_1^{i},\tilde{Z}_2^{i-1})\nonumber\\
&\leq h(\tilde{Z}_2(i))-h(\tilde{Z}_2(i)| \vec{M},Y_1^{i},\tilde{Z}_2^{i-1},Z_2^{i-1},Z_3^{i-1})\\
&= h(\tilde{Z}_2(i))-h(\tilde{Z}_2(i)| \vec{M},Z_1^{i},\tilde{Z}_2^{i-1},Z_2^{i-1},Z_3^{i-1})\\
&= h(\tilde{Z}_2(i))-h(\tilde{Z}_2(i)| \vec{M},Z_1^{i},Z_2^{i-1},Z_3^{i-1})\\
&= h(\tilde{Z}_2(i))-h(Z_2(i)| \vec{M},Z_1^{i},Z_2^{i-1},Z_3^{i-1})\\
\label{IE2}
&= h(\tilde{Z}_2(i))-h(Z_2(i)),
\end{align}
which can be shown by using similar arguments as $(c)$-$(f)$ above. By substituting  \eqref{IE1} and \eqref{IE2} in \eqref{TIE} we obtain
\begin{align}
& n(R_{21}+R_{31}+R_{32}-\varepsilon_n)\nonumber\\
& \leq \sum_{i=1}^{n} h(Y_1(i))-h(Z_1(i))+h(\tilde{Z}_2(i))-h(Z_2(i))\\
& \leq \frac{n}{2}\log(1+h_3^{2}P+h_2^{2}P)+\frac{n}{2}\log\left(1+\frac{h_1^{2}}{h_2^{2}}\right),
\end{align}
which follows since the Gaussian distribution is a differential entropy maximizer. Now by dividing by $n$ and then letting $n\to\infty$, we obtain
\begin{align}
R_{21}+R_{31}+R_{32}&\leq C(h_3^{2}P+h_2^{2}P)+C\left(\frac{h_1^{2}}{h_2^{2}}\right),
\end{align}
which proves the statement of the lemma.
\end{proof}

Now, we have a bound on the sum $R_{21}+R_{31}+R_{32}$ which scales as $\frac{1}{2}\log(P)$ as $P$ grows. Next, we provide a bound on $R_{12}+R_{13}+R_{23}$ which complements the previous bound to a sum-capacity upper bound.

\begin{lemma}
\label{Lemma:SecondBound}
An achievable rate for the 3-D2D channel must satisfy
\begin{align}
R_{12}+R_{23}+R_{13} \leq C\left(h_3^2P+h_3^2\frac{h_1^2}{h_2^2}P\right)+\frac{1}{2}.
\end{align}
\end{lemma}
\begin{proof}
The derivation of the bound is similar to that in Lemma \ref{Lemma:FirstBound} with one difference, we start with enhancing user 3 by replacing the noise $z_3$ with $\frac{h_2}{h_3}z_3$ which is weaker than $z_3$ since $h_2^{2}\leq h_3^{2}$ (cf. \eqref{Ordering}). We denote the received signal of the enhanced receiver $y_3'$. Then, we give $m_{21}$ and $\tilde{z}_2^n=z_2^n-z_3^n$ to user 3 as side information. Now, after decoding $m_{23}$, user~3 has all the information available to user 2 at time instant $i=1$, i.e., the message pair $(m_{21},m_{23})$, and it can generate the first symbol of $x_2^{n}$, i.e., $x_2(1)$. By using $y_3'(1)=h_2x_1(1)+h_1x_2(1)+\frac{h_2}{h_3}z_3(1)$, user 3 can obtain $\tilde{y}_2(1)=h_2x_1(1)+\frac{h_2}{h_3}z_3(1)$. After multiplying $\tilde{y}_2(1)$ by $\frac{h_3}{h_2}$, and  adding $h_1x_3(1)$, user 3 obtains 
\begin{align}
\tilde{y}_2(1)=h_3x_1(1)+h_1x_3(1)+z_3(1).
\end{align}
Now by adding $\tilde{y}_2(1)$ to $\tilde{z}_2(1)$, user 3 gets $y_2(1)$ and can generate $x_2(2)$. By repeating this procedure, user 3 can generate all instances of $y_2^n$, and then using $m_{21}$ and $m_{23}$, it can decode $m_{12}$ just as user 2 can decode it. Therefore, we can write 
\begin{align*}
n(R_{12}+R_{23}+R_{13}-\varepsilon_n)\leq I(\vec{M}_{1},M_{23};Y_3'^n,\tilde{Z}_2^n, \hat{\vec{M}}_{1},M_{32})
\end{align*}
by using Fano's inequality, where $\varepsilon_n\to0$ as $n\to\infty$. Now by proceeding with similar steps as those in the proof of Lemma~\ref{Lemma:FirstBound} (i.e., $(c)$ to $(f)$), we can obtain
\begin{align}
R_{12}+R_{23}+R_{13} \leq C\left(h_3^2P+h_3^2\frac{h_1^2}{h_2^2}P\right)+\frac{1}{2},
\end{align}
which is the desired upper bound.
\end{proof}

Now we have the two components necessary for establishing our sum-capacity upper bound, given in the next theorem.
\begin{theorem}
\label{Thm:SCUB}
The sum-capacity of the 3-D2D channel satisfies
\begin{align}
\label{SCUB}
C_\Sigma\leq 2C(h_3^{2}P)+2.
\end{align}
\end{theorem}
\begin{proof}
To prove this theorem, we use the upper bound on the sum $R_{21}+R_{31}+R_{32}$ given in Lemma \ref{Lemma:FirstBound} which satisfies
\begin{align}
R_{21}+R_{31}+R_{32} &\leq C(h_3^2P+h_2^2P)+C\left(\frac{h_1^{2}}{h_2^{2}}\right)\\
&\quad\leq C(2h_3^2P)+C\left(1\right)\\
&\quad= C(2h_3^2P)+\frac{1}{2},
\end{align}
since $h_3^{2}\geq h_2^{2}\geq h_1^{2}$ \eqref{Ordering}, and the upper bound on the sum $R_{12}+R_{23}+R_{13}$ given in Lemma \ref{Lemma:SecondBound} which satisfies
\begin{align}
R_{12}+R_{23}+R_{13} &\leq C\left(h_3^2P+h_3^2\frac{h_1^2}{h_2^2}P\right)+\frac{1}{2}\\
&\leq C\left(2h_3^2P\right)+\frac{1}{2}.
\end{align}
By adding the two bounds, we get
\begin{align}
R_\Sigma &\leq 2C(2h_3^2P)+1<2C(h_3^2P)+2,
\end{align}
which proves that any achievable rate tuple must have a sum that satisfies $R_\Sigma \leq 2C(h_3^2P)+2$. Therefore, we obtain the desired sum-capacity upper bound.
\end{proof}
Clearly, this sum-capacity upper bound is tighter than that obtained from the cut-set bounds as $P$ increases. Namely, this bound behaves as $\log(P)$ as $P$ grows (2 DoF), in contrast to the cut-set bounds that behave as $\frac{3}{2}\log(P)$ (3 DoF). Next, we show that this upper bound is achievable within a constant gap.

\section{Transmission Strategies}
\label{Achievability}
The derived sum-capacity upper bound in Theorem \ref{Thm:SCUB} has the following desirable structure: it is equal to twice the capacity of the strongest channel ($h_3$) plus a constant. This directly indicates a near sum-rate optimal scheme for the 3-D2D channel. Namely, by allowing the two users sharing this strongest channel to communicate, we can achieve the upper bound within a constant gap. 

Hence, let users 1 and 2 communicate via the channel $h_3$ while leaving user 3 silent. This reduces the 3-D2D channel to a two-way channel. As shown in \cite{Han}, the following rates are achievable in the resulting two-way channel
\begin{align}
R_{12}& \leq C(h_3^{2}P)\\
R_{21}& \leq C(h_3^{2}P).
\end{align}
By adding the two achievable rates, we conclude that the following sum-rate is achievable
\begin{align}
R_\Sigma& \leq 2C(h_3^{2}P).
\end{align}
Comparing this achievable sum-rate and the upper bound \eqref{SCUB} in Theorem \ref{Thm:SCUB}, we can see that this achievable sum-rate is within a gap of 2 bits of the sum-capacity upper bound. This leads to a sum-capacity characterization within a gap of 2 bits as follows
\begin{align}
2C(h_3^{2}P)\leq C_\Sigma \leq 2C(h_3^{2}P)+2.
\end{align}
This proves the main result of the paper given in Theorem \ref{Thm:Main}.

Although this transmission strategy suffices to show the achievability of the sum-capacity upper bound of the 3-D2D channel within a constant gap, we would like to additionally highlight the following interesting possibility. Consider a scenario where users 2 and 3 want to communicate with a rate that can not be supported by the channel $h_1$. Interestingly, if this rate can be supported by the channel $h_2$, then the two users can successfully communicate via user 1 as follows. Users 2 and 3 use nested-lattice codes \cite{NazerGastpar} to establish physical-layer network-coding \cite{WilsonNarayananPfisterSprintson} for bi-directional relaying via user 1. In other words, users 2 and 3 communicate via user 1 as in a two-way relay channel\footnote{Users 2 and 3 can also communicate via user 1 using quantize-forward as in \cite{AvestimehrSezginTse}}. This leads to the achievability of the rates
\begin{align}
R_{23},R_{32} \leq C\left(h_2^{2}P-\frac{1}{2}\right),
\end{align}
which is larger than the rates that can be achieved via the channel $h_1$ given by
\begin{align}
R_{23},R_{32} \leq C\left(h_1^{2}P\right),
\end{align}
as long as $h_2^{2}\geq h_1^{2}+\frac{1}{2P}$. This condition is guaranteed by \eqref{Ordering} at high $P$. While this strategy does not achieve the sum-capacity within a constant gap, it is useful for achieving high communication rates between users 2 and 3.

\section{Conclusion}
\label{Conclusion}
We studied the 3-user D2D channel (a three-way channel) and obtained its sum-capacity within a constant gap. While this required deriving a novel genie-aided upper bound, the achievability strategy is in fact simple; the sum-capacity is achievable within a constant gap by letting only two users communicate via the strongest channel. This insight is interesting since it shows that increasing the number of users in a multi-way communications channel \alert{from 2 to 3} does not increase the sum-capacity scaling behaviour of the channel. \alert{The authors believe that this conclusions extends to larger multi-way communications channels with more than 3 users.} Note that this is analogous to the multi-way relay channel where increasing the number of users also does not increase the capacity scaling behaviour. As an extension to this work, it would be interesting to find out if other analogies exist between multi-way channel and multi-way relay channels.

\bibliography{myBib}

\end{document}